\documentclass[showpacs, amsmath, amssymb,twocolumn]{revtex4}

\usepackage{amsfonts}
\usepackage{amsbsy}
\usepackage{latexsym}
\usepackage{mathrsfs}
\usepackage{bbm}
\usepackage{bm}
\usepackage[dvipdfm]{graphicx}

\newtheorem{definition}{Definition}
\newtheorem{proposition}{Proposition}
\newtheorem{theorem}{Theorem}
\newtheorem{lemma}{Lemma}

\newenvironment{proof}[1][Proof]{\textbf{#1.} }{\ \rule{0.5em}{0.5em}}

\newcommand{\be}{\begin{eqnarray}}
\newcommand{\ee}{\end{eqnarray}}
\def\({\left(}
\def\){\right)}
\def\[{\left[}
\def\]{\right]}

\def\C{\mathbb{ C}}

\newcommand{\ket}[1]{\left| #1 \right\rangle}

\newcommand{\sla}[1]{\rlap{\kern .15em /}#1}

\newcommand{\ot}{\otimes}
\newcommand{\Bot}{\bigotimes}

\newcommand{\rA}{\rightarrow}

\begin{document}

\title{Universal Separability and Entanglement in Identical Particle Systems}
\author{
 Toshihiko Sasaki$^{1}$, Tsubasa Ichikawa$^2$, and Izumi Tsutsui$^{3}$
}

\affiliation{
$^1$Photon Science Center, University of Tokyo, 7-3-1 Hongo, Bunkyo-ku, Tokyo 113-8656, Japan\\
$^2$ Department of Physics, Gakushuin University,1-5-1 Mejiro, Toshima-ku, Tokyo 171-8588, Japan\\
$^3$Theory Center, Institute of Particle and Nuclear Studies, High Energy Accelerator Research Organization (KEK), 1-1 Oho, Tsukuba, Ibaraki 305-0801, Japan
}

\begin{abstract}
Entanglement is known to be a relative notion, defined with respect to the choice of physical observables to be measured ({\it i.e.}, the measurement setup used).   This implies that, in general, the same state can be both separable and entangled for different measurement setups, but this does not exclude the existence of states which are
separable (or entangled) for all possible setups.  We show that for systems of bosonic particles there indeed exist such universally separable states: they are i.i.d.~pure states.  In contrast, there is no such state for fermionic systems with a few exceptional cases.   We also find that none of the fermionic and bosonic systems admits 
universally entangled states.  
\end{abstract}

\maketitle

\section{Introduction}

Quantum entanglement is a crucial trait of quantum mechanics: it yields correlation in measurement outcomes 
that cannot be emulated by classical (local realistic) theories \cite{Bell64}.   Arguably, 
entanglement has been one of the most important subjects of study in quantum information science over the last decades, where it serves as an indispensable resource for implementing quantum algorithms and protocols \cite{NC00}.

There are mainly two lines of thoughts to characterize the entanglement: One is a structural
description based on the formal tensor product structure of the state space of a given 
composite system. In this description, pure states are entangled if it cannot be written as 
a product state \cite{Horodecki2009}. 
The other is a phenomenological description based 
on correlations. 
Pure states are entangled if they exhibit non-trivial correlation in measurement 
outcomes of mutually distinct and simultaneously measurable physical quantities \cite{Ghirardi2002, Tichy2011}. 

These two approaches are consistent as far as we consider distinguishable particles, but apparent inconsistency emerges when we consider
identical particles 
\cite{Peres1993, Paifmmodecheckselsevsfikauskas2001, 
Li2001, Eckert2002, Gittings2002, Shi2003, Enk2003, Wiseman2003, Ghirardi2004, Zanardi2004, Grabowski2011}.  Namely,  for systems of identical particles, there are some states which 
are entangled according to the former approach while they are not according to the latter.   The trouble stems from 
the fact that, even though the quantum states of identical particles are necessarily either symmetric for bosons or antisymmetric for fermions under the exchange of particles,  this formal non-product structure does not imply nontrivial correlation in the measurement outcomes.  This is a problem that cannot be dismissed, because most of the actual systems whose states have been realized as entangled are made of identical particles, whether they be photons, electrons or some other particles.   The qubit devises which are currently envisaged to carry out quantum computation are mostly designed by means of identical particles.  

In our recent work \cite{Ichikawa2010, Sasaki2011}, we presented a convenient scheme of entanglement which 
dissolves the apparent inconsistency in the previous approaches.   The idea is that,
since the latter of the two approaches defines the entanglement relative to the 
choice of the physical quantities to be measured, or the measurement setup in short, that choice can be used to 
provide the tensor product structure needed in the former for examining the entanglement.  The judgments of entanglement 
in the two are now reconciled and made consistent even for systems of identical particles.  

Once this relative nature of entanglement is taken into account properly, the following question 
arises: are there quantum states which are non-entangled for all
possible measurement setups?   If there are, such states embody non-entanglement in an absolute sense, and 
we call them {\it universally separable}.   Conversely, if  there are quantum states which are entangled for all possible measurement setups, we may call them  {\it universally entangled}.  This question is important, not just because the answer may suggest some novel notion of intrinsic entanglement  that is independent of measurement setups, but also because it should be useful in preparing entangled states which are robust against measurement perturbations.
In fact, a partial answer to the question has been already given \cite{Sasaki2011}: independently
and identically distributed (i.i.d.)~pure states (defined in Section~\ref{Separability-and-Entanglement-Relative-to-Measurement-Setups}) are universally separable.

In this article, we provide a complete answer to this question.   We present it in three theorems.
Theorem 1, which is essentially a no-go theorem, 
states that no pure states of $N$ fermions are universally separable unless the dimensionality $n$ of the constituent system is too small $n \le N+1$ to accommodate sufficient distinctive states when it is built into the composite system.   
In contrast, Theorem 2 tells us that  all pure states $N$ bosons are universally separable for $n \le 3$, and that for $n \ge 4$ the states are universally separable if and only if 
they are i.i.d.~pure states.  
Theorem 3 then gives another no-go theorem, which shows that there are no universally entangled
states in both the fermionic and bosonic systems.

This paper is organized as follows.
In Section~\ref{Separability-and-Entanglement-Relative-to-Measurement-Setups}, we recall briefly our scheme of entanglement  \cite{Sasaki2011} and provide prerequisites
for our arguments.
Section~\ref{fermithm} deals with fermionic systems and proves Theorem 1.   Similarly, 
Section~\ref{bosonthm} deals with bosonic systems and proves Theorem 2.   Theorem 3 is then treated in 
Section~\ref{univent}.  Our conclusion and discussions are given in 
Section~\ref{conclusion}.
Section~\ref{fermithm} and Section~\ref{bosonthm} are mostly devoted to the proofs of the theorems, and readers who are uninterested in the technical details
may skip these except for the statements of the two theorems presented in the beginning of the sections.

\section{Separability and Entanglement Relative to Measurement Setups}
\label{Separability-and-Entanglement-Relative-to-Measurement-Setups}

\subsection{Preliminaries}
\label{prel}
In this section, we outline our scheme of entanglement for systems of identical particles \cite{Ichikawa2010, Sasaki2011}.  Our system of interest consists of $N$ identical particles which are either bosons or fermions.
Let ${\cal H}_i$ be the Hilbert space of the $i$th constituent particle of the system with dimension $n$: ${\cal H}_i=\C^n$
for $i=1,2,\dots,N$. 
To take account of the identical nature of the particles, we introduce the symmetrizer $\mathcal{S}$ for bosons and the antisymmetrizer $\mathcal{A}$ for fermions defined by 
\begin{equation}
 \mathcal{S}=\frac{1}{N!} \sum_{\sigma\in \mathfrak{S}_{N}} \pi_{\sigma},
 \qquad
 \mathcal{A}=\frac{1}{N!} \sum_{\sigma\in \mathfrak{S}_{N}} \text{sgn}(\sigma) \pi_\sigma.  \nonumber
  \label{symop}
\end{equation}
Here, $\mathfrak{S}_{N}$ is the symmetric group associated with the permutation
$i \to \sigma(i)$ of the particles $i = 1, \ldots, N$, which is represented by the unitary operator $\pi_\sigma$ in 
the tensor product Hilbert space $\bigotimes_{i=1}^N\mathcal{H}_i$.
More precisely, the unitary operator $\pi_\sigma$ acts on the vector in $\bigotimes_{i=1}^N\mathcal{H}_i$ as
\be
&&\pi_\sigma \ket{\psi_1}_1\ot\cdots\ot|\psi_N\rangle_N\nonumber\\
&&=|\psi_{\sigma^{-1}(1)}\rangle_1\ot\cdots\ot|\psi_{\sigma^{-1}(N)}\rangle_N,
\label{perm}
\ee
for state vectors $\ket{\psi_j}_i\in{\cal H}_i$.

In what follows, as we have done in Eq.~(\ref{perm}), we always arrange the one-particle states in any
tensor product in the increasing order of the label of the constituent Hilbert space 
from the left to the right.   With this ordering convention, one can dispense with the label which refers to the constituent Hilbert space, allowing one to write the one-particle state simply as $\ket{\psi_j}$ instead of 
$\ket{\psi_j}_i$.

The total Hilbert space of the system is thus given by 
\be
\mathcal{H}_{\mathcal{X}} = \left[{\cal H}_1\ot{\cal H}_2\ot\cdots\ot{\cal H}_N\right]_{\cal X},
\label{toth}
\ee
where we have introduced the notation,
\begin{equation}
[{\cal K}]_{\cal X}:=\{\mathcal{X}\ket{\Psi} \mid \ket{\Psi} \in \mathcal{K}\}, \nonumber
\label{prosym}
\end{equation}
which denotes the subspace of $ \mathcal{K}$ obtained by the projection $\mathcal{X} = \mathcal{S}$ for bosons or $\mathcal{X} = \mathcal{A}$ for fermions, respectively.

Next, we furnish a structure which defines the entanglement with respect to the choice of measurement setups.
This additional structure consists of two ingredients.
One is how the total system breaks into subsystems, which is taken care of a partition of the system of $N$ particles, namely, a set $\Gamma=\{\Gamma_k\}_{k=1}^s$ consisting of elements which are mutually exclusive ($\Gamma_i\cap\Gamma_j=\emptyset$ for $i\neq j$) and exhaustive
$\bigcup_{k=1}^s\Gamma_k=\{1,2,\ldots, N\}$ in the total system.
The other is how these subsystems can be separately measured, which is dealt with 
an orthogonal decomposition of the one-particle Hilbert space $\C^n$, namely, 
a set $V=\{V_k\}_{k=1}^s$ of orthogonal subspaces $V_k$ such that 
\begin{equation}
 \C^n \supseteq V_1\oplus V_2\oplus\dots\oplus V_s.\nonumber
\end{equation}
Note that we have left the possibility of the case where the direct sum $\bigoplus_{k=1}^sV_{k}$ may not comprise the entire one-particle Hilbert space $\C^n$.  The set of orthogonal subspaces is meant here for describing the situation where
the states of the $N$ particles are measured by $s$ remotely separated apparatuses labeled by $k = 1, \ldots, s$.  
If one measures the states of $|\Gamma_k|$ particles in the subset $\Gamma_k$ with the apparatus $k$ for which the subspace $V_k$ is allocated, then the corresponding Hilbert space for the subset $\Gamma_k$ reads
\be
\mathcal{H}_{\mathcal{X}}(\Gamma_k,V_k) = \left[ V_k^{\otimes |\Gamma_k|}\right]_{\mathcal{X}}.
\label{sspss}
\ee
The pairwise orthogonality of $V_k$ is important to make distinctions among the particles belonging to different subsets, which is usually fulfilled by the locality of the measurement apparatuses.    

\begin{figure}[t]
   \centering
   \includegraphics[width=3in]{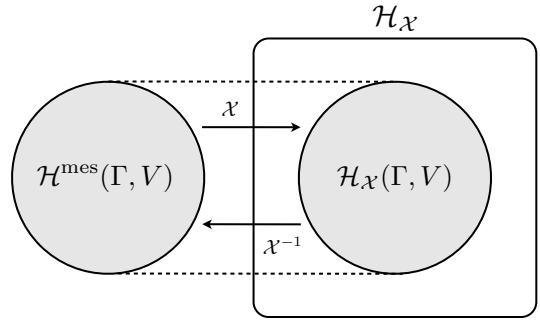} 
   \caption{Schematic diagram of spaces introduced for examination of entanglement. Given a pair $(\Gamma, V)$,  one can find in the total Hilbert space ${\cal H}_{\cal X}$
the subspace ${\cal H}_{\cal X}(\Gamma, V)$ which is unitarily equivalent to the space  ${\cal H}^{\rm mes}(\Gamma, V)$ describing the measurement results under the map ${\cal X}$.}
   \label{idmap}
 \end{figure}

By combining all the subspaces (\ref{sspss}), one can construct the Hilbert space of measurable states for the entire $N$ particles
as
\begin{equation}
\begin{split}
\mathcal{H}_{\mathcal{X}}(\Gamma,V) = \left[ \bigotimes_{k=1}^s\mathcal{H}_{\mathcal{X}}(\Gamma_k,V_k)\right]_{\mathcal{X}}.\nonumber
\end{split}
\end{equation}
Note that $\mathcal{H}_{\mathcal{X}}(\Gamma,V)$ is a subspace of $\mathcal{H}_{\mathcal{X}}$ in (\ref{toth}), but this is sufficient for our purpose because those states which belong to the orthogonal complement of $\mathcal{H}_{\mathcal{X}}(\Gamma,V)$ in $\mathcal{H}_{\mathcal{X}}$ cannot be detected by the apparatus in the measurement setup specified by the pair $(\Gamma,V)$ and hence can be safely ignored.  We also note that with
\be
\mathcal{H}^{\text{mes}}(\Gamma,V) := \bigotimes_{k=1}^s\mathcal{H}_{\cal X}(\Gamma_k,V_k),\nonumber
\ee
one can show \cite{Sasaki2011} that 
the map $\mathcal{X}:\mathcal{H}^{\text{mes}}(\Gamma,V)\rA\mathcal{H}_{\mathcal{X}}(\Gamma,V) $ is unitary, and
obviously by the inverse map $\mathcal{X}^{-1}:\mathcal{H}_{\mathcal{X}}(\Gamma,V)\rA \mathcal{H}^{\text{mes}}(\Gamma,V)$ we have
\be
\mathcal{X}\bigotimes_{k}\ket{\psi_k}\mapsto \bigotimes_{k}\ket{\psi_k},\qquad \ket{\psi_k} \in \mathcal{H}_{\cal X}(\Gamma_k,V_k),
\label{imap}
\ee
up to normalization.  This shows that the two spaces, $\mathcal{H}_{\mathcal{X}}(\Gamma,V)$ and $\mathcal{H}^{\text{mes}}(\Gamma,V)$, are isomorphic.  The point is that
the latter space $\mathcal{H}^{\text{mes}}(\Gamma,V)$ is equipped with a tensor product structure which can be used to decide entanglement of the state in $\mathcal{H}_{\mathcal{X}}(\Gamma,V)$ (see Fig.~\ref{idmap}).

To be more explicit, our procedure of examining entanglement with respect to the measurement setup specified by $(\Gamma,V)$ consists of the following four steps \cite{Sasaki2011}:
\begin{enumerate}
\item 
Given a state $\ket{\Psi} \in \mathcal{H}_\mathcal{X}$, we project it onto the subspace $\mathcal{H}_{\mathcal{X}}(\Gamma,V)$ and denote it as $\ket{\Psi(\Gamma,V)}$.
\item
We then convert the state $\ket{\Psi(\Gamma,V)}$ from $\mathcal{H}_{\mathcal{X}}(\Gamma,V)$ to $\mathcal{H}^{\text{mes}}(\Gamma,V)$ by (\ref{imap}) and denote it as $\ket{\Psi^{\text{mes}}}$.
\item
Based on the  tensor product structure of $\mathcal{H}^{\text{mes}}(\Gamma,V)$, 
we determine whether the state $\ket{\Psi^{\text{mes}}}$ is entangled or not by the standard definition of entanglement available for distinguishable particles.
\item
Since $\ket{\Psi}$ and $\ket{\Psi(\Gamma, V)}$ are indistinguishable in our measurement setup, 
we can identify the entanglement of  $\ket{\Psi(\Gamma, V)}$ with that of $\ket{\Psi}$.
\end{enumerate}
A similar argument is possible also for mixed states for which the restriction and the unitary map can be generalized straightforwardly.

\subsection{Universal Separability and i.i.d.~States}
\label{section-unisep}
As shown in Section \ref{prel}, the entanglement of the identical particle systems depends on the choice of the pair $(\Gamma,V)$, that is, how to measure the state we are given.  Thus it is curious to know, under what choice of measurement setup, a given state becomes entangled or unentangled.    Concerning this, the first question one addresses will be  if there is a special state which cannot be made entangled under any choice of measurement setup.   For this, it is convenient to introduce:
\begin{definition}
A state $\ket{\Psi}$ is universally separable  (USEP) if $\ket{\Psi(\Gamma, V)}$ is separable for any choice of $(\Gamma,V)$.
\end{definition}

A simple example of USEP states is provided by independently and identically distributed (i.i.d.) pure states, which are the states that can be written as
\begin{equation}
 \ket{\Psi} = \ket{\phi}^{\otimes N},\nonumber
\end{equation}
with some $\ket{\phi}\in \mathbb{C}^n$.
Note that, being symmetric states,  i.i.d.~states are allowed only for bosonic systems.

To examine if the i.i.d.~states are indeed USEP,  let us apply our entanglement criterion to the i.i.d.~state.
Given a pair $(\Gamma, V)$, the projected state $\ket{\Psi(\Gamma,V)}$ of the i.i.d.~state $\ket{\Psi}$
is given by
\begin{equation}
 \ket{\Psi(\Gamma,V)} =
 \sqrt{M}
 \mathcal{S}\bigotimes_{k=1}^s \ket{\phi_k}^{\otimes|\Gamma_k|},
 \label{sep}
\end{equation}
where 
$\sqrt{M}$ is a normalization constant defined through the
multinomial coefficient $M=N!/\prod_{i=1}^s|\Gamma_i|!$
and
$\ket{\phi_k}$ is the normalized state obtained by projecting $\ket{\phi}$ onto $V_k$ and rescaling it.
This state is mapped to $\mathcal{H}^{\text{mes}}(\Gamma,V)$ as
\begin{equation}
\ket{\Psi^{\text{mes}}}=
 \bigotimes_{k=1}^s \ket{\phi_k}^{\otimes|\Gamma_k|}.\nonumber
\end{equation}
We then find that an i.i.d.~pure state  is a separable state for any $(\Gamma,V)$,   
implying that they are USEP as announced.  

Do these i.i.d.~pure states exhaust all possible USEP states in the bosonic case?  How about the USEP in the fermionic case?  These are the questions we address and answer in the following sections.  We begin by the fermionic case first, as it is simpler.

\section{Universally Separable States in Fermionic Systems}
\label{fermithm}

We first consider fermionic systems to pin down what USEP states are.   Specifically, we
prove:
\begin{theorem}
\label{th1}
For $N$-partite fermionic systems for which the constituent Hilbert space is $\mathbb{C}^n$, we have
\begin{enumerate}
 \item \label{unisep-fermionic-N1}
For $n \leq N+1$, all pure states are USEP.
 \item \label{unisep-fermionic-N2}
For $n \geq N+2$, no pure states are USEP.
\end{enumerate}
\end{theorem}

Before proceeding, we outline the basic idea of our proof.
The first statement, hereafter called case 1, is more or less trivial and straightforward to prove.  For this we just
show that the dimension of $\mathcal{H}_{\mathcal{A}}(\Gamma,V)$, 
which is regarded as the total Hilbert space under the measurement setup $(\Gamma,V)$,
is too small to accommodate any entangled states.
In contrast, the second statement (case 2) is quite nontrivial and important.
According to the definition of USEP, such a state $\ket{\Psi(\Gamma,V)}$ must be separable for arbitrary $(\Gamma,V)$.
With each $(\Gamma,V)$ chosen, this leads to some restrictions to the state $\ket{\Psi}$, and by varying the choice of 
$(\Gamma, V)$ we can show that $\ket{\Psi}=0$,
completing the proof.

\subsection{Case \ref{unisep-fermionic-N1}: $n \leq N+1$}

\noindent
\begin{proof}[Proof of Case \ref{unisep-fermionic-N1} of Theorem \ref{th1}]
To begin with, we note that since two fermions cannot occupy an identical state, we have $\dim V_i\ge|\Gamma_i|$
for all $i$, which leads to
\be
 N= \sum_{i=1}^s|\Gamma_i| \leq \sum_{i=1}^s\dim V_i \le n.\nonumber
\ee
From the condition $n \leq N+1$ of the present case, we observe $\dim V_i=|\Gamma_i|$ for all $i$ except possibly for one element.
We can label $i = s$ for such an exceptional subsystem without losing generality.  
This implies that $\dim\mathcal{H}_{\mathcal{A}}(\Gamma_i,V_i)=1$ at least for $i\in\{1,2,\ldots,s-1\}$.  The orthogonality $V_{i}\perp V_j$ for $i\neq j$ allows us
to choose an orthonormal basis $\{\ket{e_i}\}_{i=1}^n$ in ${\cal H}_{\rm c}$ such
that
\be
V_i&=&{\rm span}\{\ket{e_{\alpha_i+1}},\ldots,|e_{\alpha_i+|\Gamma_i|}\rangle\},\nonumber
\ee
where we have introduced $\alpha_i$ which are recursively defined by 
$\alpha_{i+1}=\alpha_i+|\Gamma_i|$ with the initial condition $\alpha_1=0$.
Then, by construction, $\mathcal{H}_{\mathcal{A}}(\Gamma_i,V_i)$ contains only a single state,
\be
\ket{\phi_i}=\sqrt{|\Gamma_i|!}\, {\cal A} (|e_{\alpha_i+1}\rangle\otimes\cdots\ot|e_{\alpha_i+|\Gamma_i|}\rangle)\nonumber
\ee
for $i\in\{1,2,\ldots,s-1\}$, up to an overall constant.
It follows that any state $\ket{\Psi(\Gamma,V)}$ takes the form of a separable state,
\begin{equation}
 \ket{\Psi(\Gamma,V)}=\sqrt{M}\mathcal{A}\(\ket{\phi_1}\otimes\cdots\otimes\ket{\phi_{s-1}}\otimes\ket{\psi_s}\)\nonumber
\end{equation}
with some state $\ket{\psi_s}\in \mathcal{H}_{\mathcal{A}}(\Gamma_s,V_s)$.
Since this argument is independent of the choice of $(\Gamma,V)$, we learn that $\ket{\Psi}$ is USEP.
\end{proof}

\subsection{Case \ref{unisep-fermionic-N2}: $n \geq N+2$}

To prove case \ref{unisep-fermionic-N2}, we first show:
\begin{lemma}
 \label{unisep-fermionic-N2-lemma}
 For $N=2$ and $n\geq N+2=4$,  there exists no fermionic USEP state.
\end{lemma}
\begin{proof}
We prove this by reductio ad absurdum.
Suppose that a state $\ket{\Psi}\in \mathcal{H}_{\mathcal{A}}$ is USEP.
On the one hand, by using an orthonormal basis $\{\ket{e_i}\}_{i=1}^n$ of $\mathbb{C}^n$, 
$\ket{\Psi}$ is written as
\begin{equation}
 \ket{\Psi} = \frac{1}{2}\sum_{ i,j = 1,\cdots, n}\Psi_{ij} \ket{e_i}\ket{e_j},
 \qquad 
 \Psi_{ji} = -\Psi_{ij}.\nonumber
\end{equation}
On the other hand, it follows from Slater decomposition \cite{Schliemann2001, Eckert2002} 
that $\ket{\Psi}$ is written as
\begin{equation}
 \ket{\Psi} = \sum_{i=1}^{K} z_i\mathcal{A}\ket{e_{2i-1}'}\ket{e_{2i}'},
 \qquad
  2K \leq n\nonumber
\end{equation}
by using complex number $z_i\in\C$ and an appropriate orthonormal basis $\{\ket{e_i^\prime}\}_{i=1}^n$.

Now, we choose $\Gamma$ and $V$ such that
$
 \Gamma = \{\Gamma_1,\Gamma_2\}
$
with
\be
  \Gamma_{1}=\{1\}
  \quad
  {\rm and}
  \quad
  \Gamma_2=\{2\},\nonumber
\ee
and
$
 V = \{V_1,V_2\}
$
with
\be 
V_1 = \text{span}\{\ket{e_1'},\ket{e_3'}\}
\quad
  {\rm and}
  \quad
 V_2 = \text{span}\{\ket{e_2'},\ket{e_4'}\}.\nonumber
\ee
Then we find
\begin{equation}
 \ket{\Psi(\Gamma,V)} = \kappa{\cal A}\(z_1|e^\prime_1\rangle|e^\prime_2\rangle+z_2|e^\prime_3\rangle|e^\prime_4\rangle\).\nonumber
\end{equation}
where $\kappa=\sqrt{2/(|z_1|^2+|z_2|^2)}$ is the normalization constant.
Since $\ket{\Psi}$ is supposed to be USEP, we obtain $z_1=0$ or $z_2=0$.

When $z_2=0$, we have
\begin{equation}
 \ket{\Psi (\Gamma, V)} = z_1\mathcal{A}\ket{e_1'}\ket{e_2'}.\nonumber
\end{equation}
Next, we choose $V'= \{V'_1,V'_2\}$ such that
\be
 V'_1 = \text{span}\{\ket{e_1''},\ket{e_2''}\}
 \quad
  {\rm and}
  \quad
V'_2 = \text{span}\{\ket{e_3''},\ket{e_4''}\},\nonumber
\ee
where
\be
 \ket{e_1''} = \(\ket{e_1'}+\ket{e_3'}\)/\sqrt{2},
 \qquad
 \ket{e_2''} = \(\ket{e_2'}+\ket{e_4'}\)/\sqrt{2},\nonumber\\
\ket{e_3''} = \(\ket{e_1'}-\ket{e_3'}\)/\sqrt{2},
\qquad
 \ket{e_4''} = \(\ket{e_2'}-\ket{e_4'}\)/\sqrt{2}.\nonumber
\ee
By using ${\cal A}\ket{e_3''}\ket{e_2''}=-{\cal A}\ket{e_2''}\ket{e_3''}$,  $\ket{\Psi}$ is given by
\begin{eqnarray}
 \ket{\Psi} &=& \frac{z_1}{2}\mathcal{A}
	\left(
	 \ket{e_1''}\ket{e_4''}-\ket{e_2''}\ket{e_3''}
	+
		\ket{e_1''}\ket{e_2''}+\ket{e_3''}\ket{e_4''}
	 \right)\nonumber\\
	 &+&\sum_{i=3}^{K} z_i\mathcal{A}\ket{e_{2i-1}'}\ket{e_{2i}'}.\nonumber
\end{eqnarray}
We then find
\begin{equation}
 \ket{\Psi(\Gamma,V')} = \frac{z_1}{\sqrt{2}}\mathcal{A}
	 \(\ket{e_1''}\ket{e_4''}-\ket{e_2''}\ket{e_3''}\),\nonumber
\end{equation}
which is an entangled state.
This contradicts the universal separability of $\ket{\Psi}$. 
The same argument holds when $z_1=0$.
\end{proof}

Armed with this lemma, we now finish our proof of Theorem \ref{th1}.\\
\begin{proof}[Proof of Case \ref{unisep-fermionic-N2} of Theorem \ref{th1}]
We again prove this by reductio ad absurdum.
Suppose that a normalized state $\ket{\Psi}\in{\cal H}_{\cal A}$ is USEP.
By using an orthonormal basis $\{\ket{e_i}\}_{i=1}^n$ of $\mathbb{C}^n$,
$\ket{\Psi}$ is written as
\begin{equation}
 \ket{\Psi} = \sum_{1\leq i_1< i_2 <\cdots < i_N\leq n}\Psi_{i_1i_2\cdots i_N}
	\mathcal{A}\bigotimes_{k=1}^N \ket{e_{i_k}}.
	\label{ferexp}
\end{equation}
Without loss of generality, we set $\Psi_{12\cdots N}\neq 0$ by renaming the basis vectors.
Further we set $\Gamma=\{\Gamma_1, \Gamma_2\}$ and $V=\{V_1, V_2\}$ to
\be
\Gamma_1 = \{1,2\}
\quad
{\rm and}
\quad
\Gamma_2 = \{3,4,\cdots, N\},\nonumber
\ee
and
\be
 V_2 = \text{span}\{\ket{e_i}\}_{i=1}^{N-2}
 \quad
{\rm and}
\quad
V_1 = V_2^{\perp},\nonumber
\ee
respectively. From the universal separability, $\ket{\Psi(\Gamma,V)}$ must be of the form,
\be
 \ket{\Psi(\Gamma,V)} = \sqrt{M}\mathcal{A}\(\ket{\psi_1}_{\Gamma_1}\ot\ket{\psi_2}_{\Gamma_2}\),\nonumber
\ee
where $\sqrt{M}$ is the normalization constant and $\ket{\psi_i}_{\Gamma_i}\in{\cal H}_{\cal A}(\Gamma_i, V_i)$.
Since $\dim V_1 = n - (N-2)\geq 4$, we can apply Lemma \ref{unisep-fermionic-N2-lemma} to $\ket{\psi_1}_{\Gamma_1}$ to see that there exists 
a pair $(\Gamma', V')$ with
\be
\Gamma'=\{\{1\},\{2\}\},\quad V'=\{V'_1,V'_2\}\nonumber
\ee
with which $\ket{\psi_1(\Gamma',V')}_{\Gamma_1}$ becomes an entangled state.
This implies that by the choice of the pair $(\Gamma'', V'')$ with
\be
\Gamma''=\{\{1\},\{2\},\{3,4,\cdots,N\}\}, \, \, V''=\{V'_1,V'_2,V_2\},\nonumber
\ee
the state $\ket{\Psi(\Gamma'',V'')}$ becomes entangled.
This invalidates the assumption made at the beginning. 
\end{proof}

\section{Universally separable states in bosonic systems}
\label{bosonthm}
Now we return to the bosonic systems and consider whether the converse of
the statement in Section~\ref{section-unisep} holds, that is, whether the universal
separability implies the i.i.d.~property.  We shall see that this is indeed the case except for
systems with $n \leq 3$.
\begin{theorem}
\label{th2}
For $N$-partite bosonic systems for which the constituent Hilbert space is $\mathbb{C}^n$, we have
\begin{enumerate}
 \item \label{unisep-bosonic-n3}
For $n \leq 3$, all pure states are USEP.
 \item \label{unisep-bosonic-n4}
For $n \geq 4$, pure states are USEP if and only if they are i.i.d.~pure states.
\end{enumerate}
\end{theorem}

In what follows, we shall prove this in a way similar to Section~\ref{fermithm}. 
For case 1, we will show that for any choice of $(\Gamma, V)$, the state $\ket{\Psi(\Gamma, V)}$ 
takes the form of Eq.~(\ref{sep}), meaning the separability. The proof of case 2 is technically involved and requires several lemmata and propositions before completing it.
Basically, the argument consists 
of three steps. In the first step, we show that we can choose an appropriate basis on 
which all the coefficients of a four-partite state become nonzero. In the second step, 
we prove that a USEP state is always an i.i.d.~state when $n=4$. In the last step, 
we extend this result to the general cases.

\subsection{Case \ref{unisep-bosonic-n3}: $n \leq 3$}

\noindent
\begin{proof}[Proof of Case \ref{unisep-bosonic-n3} of Theorem \ref{th2}]
We consider whether a bosonic state $\ket{\Psi}$ is separable into $s$ subsystems.
Since $\dim V_i$ is no less than 1 for $i\in\{1,2,\ldots,s\}$, the dimension of the constituent space, $n = {\rm dim}\,{\cal H}_i$, must be no less than $s$.
Besides, the number of the subsystems must be $s\geq 2$ to allow for entanglement between the subsystems.
Obviously, if $n=s = 2$ or 3, then $\dim V_i=1$ for all $i$, and if 
$n=3$, $s=2$, then $\dim V_i=1$ except for one subsystem.  

In the former case, we can write $V_i$ as $\text{span}\{\ket{e_i}\}$ by choosing appropriate $\{\ket{e_i}\}$.
It is clear that the only nonzero state in $\mathcal{H}_{\mathcal{S}}(\Gamma,V)$ is $\sqrt{M}\mathcal{S}\bigotimes_i\ket{e_i}^{\otimes |\Gamma_i|}$.
This state is separable.

In the latter case, as we did before we 
let $i=2$ be this exceptional subsystem without loss of generality.
Namely, we assume $\dim V_1 =1, \dim V_2 =2$.
We write $V_1$ and $V_2$ as $\text{span}\{\ket{e_1}\}$ and $\text{span}\{\ket{e_2},\ket{e_3}\}$, respectively.
Then by construction, with some $\ket{\psi}\in \mathcal{H}_{\mathcal{S}}(\Gamma_2,V_2)$,  any state $\ket{\Psi(\Gamma,V)}$ takes the form,
\begin{equation}
 \ket{\Psi(\Gamma,V)} =\sqrt{M}\mathcal{S}(\ket{e_1}^{\otimes |\Gamma_1|}\otimes\ket{\psi}),\nonumber
\end{equation}
which is clearly separable with respect to $(\Gamma, V)$.
Since this argument is independent of the choice of $(\Gamma,V)$, we see that
all pure states are USEP.  
 \end{proof}
 
In passing we mention that the above argument can actually be employed to prove the statement even for $n > 3$ if $n=s$ or $n=s+1$.   

\subsection{Case \ref{unisep-bosonic-n4}: $n \geq 4$}

We have already proven that all i.i.d.~pure states are USEP in Section \ref{section-unisep}.
Here we show the converse: any USEP state is an i.i.d.~state. As mentioned earlier, this
proof is composed of three steps.

First, we consider the case $n=4$.
Denoting a basis of $\mathbb{C}^4$ by $\{\ket{e_i}\}_{i=1}^4$, any state $\ket{\Psi}\in \mathcal{H}_{S}$ can be written as
\begin{eqnarray}
\label{unisep-psi-basic-notation}
 \ket{\Psi} = \sum_{l_1, l_2, l_3, l_4}\Psi_{l_1,l_2,l_3,l_4}
		\mathcal{S}\ket{e_1}^{\otimes l_1}\ket{e_2}^{\otimes l_2}\ket{e_3}^{\otimes l_3}\ket{e_4}^{\otimes l_4},\nonumber
\end{eqnarray}
where $\Psi_{l_1,l_2,l_3,l_4}\in \mathbb{C}$ and the summation is subject to the condition,
\be
l_1+l_2+l_3+l_4=N.
\label{conservation}
\ee
We then wish to show:
\begin{lemma}
 \label{unisep-lemma-nonzero-main}
Given a USEP state $\ket{\Psi}$, there exists a basis of $\mathbb{C}^4$ such that
 $\Psi_{l_1,l_2,l_3,l_4}\neq 0$ for all $l_1,l_2,l_3,l_4$ in Eq. (\ref{unisep-psi-basic-notation}).
\end{lemma}
To prove this, we need the following two lemmata.
\begin{lemma}
 \label{unisep-bosonic-n4-lemma-ab}
If $\ket{\Psi}$ is USEP, there exist two complex numbers $a_{l_1,l_2}$ and $b_{l_3,l_4}$ which satisfy
 \begin{equation*}
	\Psi_{l_1,l_2,l_3,l_4} = a_{l_1,l_2}b_{l_3,l_4}
 \end{equation*}
\end{lemma}
\begin{proof}
We choose $V= \{V_1,V_2\}$ such that
\be
V_1 = \text{span}\{\ket{e_1},\ket{e_2}\}
\quad
{\rm and}
\quad
 V_2 = {V_1}^\perp
 \label{unisep-lemma-ab-V1V2}
\ee
 and $\Gamma(i)= \{\Gamma_1(i),\Gamma_2(i)\}$ such that
\be
	\Gamma_1(i)= \{1, \cdots, i\}
	\quad
	{\rm and}
	\quad
	\Gamma_2(i)= \{i+1, \cdots,N\}
 \label{unisep-lemma-ab-Gamma-i}
\ee
for $i=1, 2, \ldots, N-1$.
Since $\ket{\Psi}$ is USEP, its observable part $\ket{\Psi(\Gamma(i),V)}$ must be separable.
It follows that there exist $a(i)_{l_1,l_2}$ and $b(i)_{l_3,l_4}$ with which we have
\begin{eqnarray}
 &&\ket{\Psi(\Gamma(i),V)}\nonumber\\
 &&=\sum_{l_1, l_2, l_3, l_4}a(i)_{l_1,l_2}b(i)_{l_3,l_4}\mathcal{S}\ket{e_1}^{\otimes l_1}\ket{e_2}^{\otimes l_2}\ket{e_3}^{\otimes l_3}\ket{e_4}^{\otimes l_4},\nonumber
\end{eqnarray}
where the summation is subject to the conditions,
\be
l_1+l_2=i,
\qquad
{\rm and}
\qquad
l_3+l_4=N-i.
\label{sumcond}
\ee
On the other hand, since $\mathbb{C}^4=V_1\oplus V_2$, we have
\begin{equation}
	\mathcal{H}_{\mathcal{S}} =\left[\(V_1\oplus V_2\)^{\otimes N}\right]_{\mathcal{S}}
	= V_1^{\otimes N} \oplus \bigoplus_{i=1}^{N-1} \mathcal{H}_{\mathcal{S}}(\Gamma(i),V)
	\oplus V_2^{\otimes N}.\nonumber
\end{equation}
This means that any USEP state $\ket{\Psi}\in \mathcal{H}_{\mathcal{S}}$ takes the form
\begin{eqnarray}
 \label{unisep-psi-basic-notation2}
 \ket{\Psi} &=&
 \ket{\Psi'}+\sum_{i=1}^{N-1}\ket{\Psi(\Gamma(i),V)}+\ket{\Psi''},
 \end{eqnarray}
 where
 \be
 &&\ket{\Psi^\prime}=\sum_{l_1+l_2=N} a(N)_{l_1,l_2} {\cal S}\ket{e_1}^{\otimes l_1}\ket{e_2}^{\otimes l_2},\nonumber\\
 &&\ket{\Psi^{\prime\prime}}=\sum_{l_3+l_4=N} b(0)_{l_3,l_4} {\cal S}\ket{e_3}^{\otimes l_3}\ket{e_4}^{\otimes l_4}.\nonumber
\ee
Note that the states appearing in the RHS of Eq.~(\ref{unisep-psi-basic-notation2}) are not normalized, and hereafter we shall not necessarily be concerned with
normalization for simplicity.

By introducing formal coefficients
\be
b(N)_{l_3,l_4}=a(0)_{l_1,l_2}=1,\nonumber
\ee
Eq.~(\ref{unisep-psi-basic-notation2}) can be rewritten as 
 \begin{eqnarray}
 \ket{\Psi}=\sum_{i=0}^N\sum_{l_1, l_2, l_3, l_4}&&a(i)_{l_1,l_2}b(i)_{l_3,l_4}\nonumber\\
 &&\times \, \mathcal{S}\ket{e_1}^{\otimes l_1}\ket{e_2}^{\otimes l_2}\ket{e_3}^{\otimes l_3}\ket{e_4}^{\otimes l_4}\nonumber
\end{eqnarray}
with the summation condition (\ref{sumcond}),
which implies the statement of Lemma  \ref{unisep-bosonic-n4-lemma-ab}.
\end{proof}

When we change Eq.(\ref{unisep-lemma-ab-V1V2}) in Lemma \ref{unisep-bosonic-n4-lemma-ab} as
\begin{equation}
 V_1 \rightarrow \text{span}\{\ket{e_1},\ket{e_3}\},\;
 V_2 \rightarrow \text{span}\{\ket{e_2},\ket{e_4}\},\nonumber
\end{equation}
we obtain $c_{l_1,l_3}$ and $d_{l_2,l_4}$, such that
\begin{equation}
\label{unisep-lemma-nonzero-cd}
 \Psi_{l_1,l_2,l_3,l_4} = c_{l_1,l_3}d_{l_2,l_4}.
\end{equation}

The next lemma ensures that $\Psi_{l_1,l_2,l_3,l_4}\neq 0$.
\begin{lemma}
\label{unisep-lemma-nonzero-basic}
Consider the set of states $\ket{\Phi_p}\neq0$ for $p = 1, \ldots, N$ given by
 \begin{equation*}
	\ket{\Phi_p} = \sum_{l_1+l_2=p} a_{l_1,l_2} \,\mathcal{S}\ket{e_1}^{\otimes l_1}\ket{e_2}^{\otimes l_2}.
 \end{equation*}
Then there exists a unitary transformation $U\in {\rm U}(2)\subset{\rm U}(4)$ such that the states $\ket{\Phi_p}$
become
\begin{equation*}
 \ket{\Phi_p} = \sum_{l_1+l_2=p} a'_{l_1,l_2} \,\mathcal{S}\ket{e'_1}^{\otimes l_1}\ket{e'_2}^{\otimes l_2},
 \quad
a_{l_1,l_2}^\prime\neq0,
\,\,\,
\forall l_1, l_2,\nonumber
\end{equation*}
with
\begin{equation*}
 \ket{e'_1}=U\ket{e_1},
 \qquad
 \ket{e'_2} = U\ket{e_2}
\end{equation*}
and
\be
{\rm span}\{\ket{e_1},\ket{e_2}\}={\rm span}\{\ket{e_1^\prime}, \ket{e_2^\prime}\}.\nonumber
\ee
\end{lemma}
\begin{proof}
Let us parameterize $U$ as
\begin{equation*}
U^{-1} =
\begin{pmatrix}
 \xi& \eta\\
 -\eta^*& \xi^*
\end{pmatrix},
\qquad
\xi, \eta\in \mathbb{C},
\quad
 |\xi|^2+|\eta|^2=1,
\end{equation*}
and express $a'_{l_1,l_2}$ explicitly in terms of the parameters as
\begin{eqnarray}
 a'_{t-(k+l),k+l} =  \sum_{l_1,l_2} a_{l_1,l_2}
{l_1\choose k}
{l_2 \choose l}
\xi^k(\xi^*)^{l_2}\eta^l(-\eta^*)^{l_1-k},
\nonumber
\end{eqnarray}
where 
${i\choose j}$ is the binomial coefficient and $0\le k\le l_1$, $0\le l\le l_2$.
Since $\ket{\Psi_p} \neq 0$, there exists a doublet $(l_1, l_2)$ such that $a_{l_1,l_2}\neq 0$.
Then we may interpret $a'_{l_1,l_2}$ as a polynomial of $\xi, \xi^*, \eta, \eta^*$ with a finite degree.
The dimension of the parameter space of $\xi$ and $\eta$ which satisfies  $a'_{l_1,l_2}= 0$ is less than the original one.
Because of their dimensionality, the union of the parameter spaces with $a'_{l_1,l_2}= 0$ cannot cover the original one.
This means that  there always exists a pair such that $a'_{l_1,l_2}\neq 0$ for all $l_1, l_2$ simultaneously. 
\end{proof}

Now we have:

\noindent
\begin{proof}[Proof of Lemma \ref{unisep-lemma-nonzero-main}]
Since $\ket{\Psi}$ is nonzero, there exists $\bar{l}_1,\bar{l}_2, \bar{l}_3,\bar{l}_4$ such that $\Psi_{\bar{l}_1,\bar{l}_2, \bar{l}_3,\bar{l}_4} \neq 0$. 
Then, from Lemma~\ref{unisep-bosonic-n4-lemma-ab} and the universal separability of $\ket{\Psi}$,
we find $a_{\bar{l}_1,\bar{l}_2}\neq 0$ and $b_{\bar{l}_3,\bar{l}_4}\neq 0$.
According to Lemma \ref{unisep-lemma-nonzero-basic}, we can choose a basis in which the following coefficients become non-vanishing,
\begin{equation}
 \label{unisep-lemma-cond-kl}
 a^\prime_{L_1-k,k}\neq0,
 \qquad
 b^\prime_{L_2-l,l}\neq0,\nonumber
 \end{equation}
 where we have introduced $L_1=\bar{l}_1+\bar{l}_2$ and $L_2=\bar{l}_3+\bar{l}_4$. Note that
 $
 0\leq k \leq L_1
 $
 and
 $
 0\leq l \leq L_2.
$

Further, from Eq.~(\ref{unisep-lemma-nonzero-cd}),
we find
\be
c_{L_1-k, L_2-l}^\prime d_{k,l}^\prime= a^\prime_{L_1-k,k} b^\prime_{L_2-l,l}\neq0,\nonumber
\ee
which means that for all $k$ and $l$, there exists at least one nonzero term in the following state;
\begin{eqnarray}
\sum_{k,l} c_{L_1-k,L_2-l}^\prime d_{k,l}^\prime {\cal S}\ket{e_1^\prime}^{\otimes (L_1-k)}\ket{e_2^\prime}^{\ot k}\ket{e_3^\prime}^{\otimes (L_2-l)}\ket{e_4^\prime}^{\ot l}. \nonumber
\end{eqnarray}
Recall that Lemma \ref{unisep-lemma-nonzero-basic} assures us a basis in which we have
\begin{eqnarray}
 c_{N-k-l-p,p}^{\prime\prime}\neq0,
 \qquad
 d_{k+l-q,q}^{\prime\prime}\neq 0. \nonumber
\end{eqnarray}
Since there are three independent parameters $k, l, m$,  the four indices of $c''_{l_1, l_3}\neq0$ and  $d''_{l_2,l_4}\neq0$ freely run from 0 to $N$ under the condition (\ref{conservation}).
Thus, by working with this basis, all components $\Psi''_{l_1,l_2,l_3,l_4}$ are nonzero, which completes the proof.
\end{proof}

According to Lemma  \ref{unisep-lemma-nonzero-main},
no generality is lost by assuming $\Psi_{l_1,l_2,l_3,l_4}\neq 0$ for all $\{l_i\}_{i=1}^4$.
Now, we proceed to the second step, where we shall show the following proposition.
\begin{proposition}
 For $n=4$, a bosonic pure state is USEP if and only if it is an i.i.d.~pure state.
 \label{n4iid}
\end{proposition}
To prove this, we need:

\begin{lemma}
 \label{symmetry-general-lemma-binominal-12-34}
For $n=4$,  any USEP state $\ket{\Psi}$ can be written as
\begin{eqnarray}
 \ket{\Psi} &=& \sum_{l_1+l_3=N}	a_{l_1,0}b_{l_3,0}\,\mathcal{S}\(\ket{e_1}+\frac{a_{0,1}}{a_{1,0}}\ket{e_2}\)^{\otimes l_1}\nonumber\\
&&\hspace{2cm}\otimes \(\ket{e_3}+\frac{b_{0,1}}{b_{1,0}}\ket{e_4}\)^{\otimes l_3}.
\label{4dec}
\end{eqnarray}
\end{lemma}
\begin{proof}
Using $U\in {\rm U}(4)$, we can construct a family of orthogonal subspaces $V'(U)=\{V'_1(U),V'_2(U)\}$ such that
\begin{equation}
	V'_1(U) = {\rm span}\{U\ket{e_1}, U\ket{e_2}\},
	\qquad
	V'_2(U)={V_1^\prime(U)}^\perp.\nonumber
\end{equation}
From the universal separability, $\ket{\Psi}$ is separable under $(\Gamma(i),V'(U))$ for any $i$ and $U$,
where $\Gamma(i)$ is that defined in Eq. (\ref{unisep-lemma-ab-Gamma-i}).
The operator $U\in{\rm U}(4)$ can be parameterized as
\begin{equation*}
U = \exp \left[ i\sum_{1\leq k\leq l \leq 4}\(\epsilon_{kl} M_{kl}+i \epsilon'_{kl}M'_{kl}\)\right],
\end{equation*}
with $M_{kl}$ and $M_{kl}^\prime$ being generators whose components are given by
\be
&& \(M_{kl} \)_{\alpha\beta} = \delta_{k\alpha}\delta_{l\beta} + \delta_{k\beta}\delta_{l\alpha},\nonumber\\
 && \(M'_{kl} \)_{\alpha\beta} = \delta_{k\alpha}\delta_{l\beta} - \delta_{k\beta}\delta_{l\alpha},\nonumber
\ee
where $\delta_{k\alpha}$ is the Kronecker delta.
Setting $\epsilon_{ij}=\epsilon'_{ij}=0$ except infinitesimal $\epsilon_{13}$ and $\epsilon'_{13}$, we obtain
\be
 &&\ket{e_1}\rightarrow \ket{e'_1} = \ket{e_1} + i\epsilon^* \ket{e_3},\nonumber\\
  &&\ket{e_3}\rightarrow \ket{e'_3} =i \epsilon \ket{e_1} +\ket{e_3},\nonumber
\ee
where we introduced $\epsilon=\epsilon_{13}+i\epsilon_{13}^\prime$ and its complex conjugate $\epsilon^*$. 
In terms of the old components $\Psi_{l_1,l_2,l_3,l_4}=a_{l_1,l_2}b_{l_3,l_4}$, the new components $\Psi^\prime_{l_1,l_2,l_3,l_4}$ are rewritten as
\be
&&\Psi^\prime_{l_1,l_2,l_3,l_4} =\Psi_{l_1,l_2,l_3,l_4}\nonumber\\
&&\qquad\qquad\ \, \times\[1-i\epsilon^*(l_3+1)\frac{A_{l_1,l_2}}{B_{l_3+1, l_4}}-i\epsilon (l_1+1)\frac{B_{l_3,l_4}}{A_{l_1+1, l_2}}\]\nonumber\\
&&\qquad\qquad\ \, +\, {\cal O}(|\epsilon|^2),\nonumber
\ee
where
\be
A_{l_1,l_2}=a_{l_1-1,l_2}/a_{l_1,l_2},
\qquad
B_{l_3,l_4}=b_{l_3-1,l_4}/b_{l_3,l_4},\nonumber
\ee
with $a_{-1,l_2}=b_{-1,l_4}=a_{N+1,l_2}=b_{N+1,l_4}=0$.
On the other hand, it follows from Lemma \ref{unisep-bosonic-n4-lemma-ab} and the universal separability of $\ket{\Psi}$ that there exist $a'_{l_1,l_2}$ and $b'_{l_3,l_4}$, such that
\begin{equation*}
\Psi^\prime_{l_1,l_2,l_3,l_4} = a_{l_1,l_2}^\prime b_{l_3,l_4}^\prime.
\end{equation*}
Hence, the ratio $\Psi^\prime_{l_1,l_2,l_3,l_4}/\Psi^\prime_{l_1^\prime,l_2^\prime,l_3,l_4}$ must be independent of $l_3$ and $l_4$.
We can rewrite this ratio as
\be
\frac{\Psi^\prime_{l_1,l_2,l_3,l_4}}{\Psi^\prime_{l_1^\prime,l_2^\prime,l_3,l_4}}=\frac{a_{l_1,l_2}}{a_{l_1^\prime,l_2^\prime}}\(1-i\epsilon^*X-i\epsilon Y\)+{\cal O}(|\epsilon|^2),\nonumber
\ee
where
\be
&&X=\frac{l_3+1}{B_{l_3+1,l_4}}\(A_{l_1,l_2}-A_{l_1^\prime,l_2^\prime}\),\nonumber\\
&&Y=B_{l_3,l_4}\(\frac{l_1+1}{A_{l_1+1,l_2}}-\frac{l_1^\prime+1}{A_{l_1^\prime+1,l_2^\prime}}\).\nonumber
\ee
Since X and Y are independent of $l_3$ and $l_4$, we obtain
\be
A_{l_1,l_2}=A_{l_1^\prime,l_2^\prime},
\qquad
\frac{l_1+1}{A_{l_1+1,l_2}}=\frac{l_1^\prime+1}{A_{l_1^\prime+1,l_2^\prime}}.
\label{recA}
\ee
The similar argument for infinitesimal $\epsilon_{24}$ and $\epsilon_{24}^\prime$ gives
\be
C_{l_1,l_2}=C_{l_1^\prime,l_2^\prime},
\qquad
\frac{l_2+1}{C_{l_1,l_2+1}}=\frac{l_2^\prime+1}{C_{l_1^\prime,l_2^\prime+1}},
\label{recC}
\ee
where
\be
C_{l_1,l_2}=a_{l_1,l_2-1}/a_{l_1,l_2}.\nonumber
\ee
Setting $l_1^\prime=l_1-1$ and $l_2^\prime=l_2+1$ in Eq.~(\ref{recA}), we find
\be
\frac{a_{l_1,l_2+1}}{a_{l_1+1,l_2}}=\frac{l_1+1}{l_1}\frac{a_{l_1-1,l_2+1}}{a_{l_1,l_2}}=\cdots=(l_1+1)\frac{a_{0,l_2+1}}{a_{1,l_2}}.\nonumber
\ee
Similarly, setting $l_1^\prime=l_1+1$ and $l_2^\prime=l_2-1$ in Eq.(\ref{recC}), we observe
\be
\frac{a_{l_1,l_2+1}}{a_{l_1+1,l_2}}=\frac{l_2}{l_2+1}\frac{a_{l_1,l_2}}{a_{l_1+1,l_2-1}}=\cdots=\frac{1}{l_2+1}\frac{a_{l_1,1}}{a_{l_1+1,0}}.\nonumber
\ee
Combining these two, we find
\be
\frac{a_{l_1,l_2+1}}{a_{l_1+1,l_2}}=(l_1+1)\frac{a_{0,l_2+1}}{a_{1,l_2}}=\frac{l_1+1}{l_2+1}\frac{a_{0,1}}{a_{1,0}}.\nonumber
\ee
We solve this recursion relation to obtain
\begin{eqnarray}
 a_{l_1,l_2} 
&=& a_{l_1+1,l_2-1} \frac{l_1+1}{l_2}\frac{a_{0,1}}{a_{1,0}}\nonumber\\
&=&\cdots = a_{l_1+l_2, 0}
{l_1+l_2 \choose l_2}
\(\frac{a_{0,1}}{a_{1,0}}\)^{l_2}.\nonumber
\end{eqnarray}
Hence, we arrive at
\be
 &&\sum_{l_1,l_2}a_{l_1,l_2} \,\mathcal{S}\ket{e_1}^{\otimes l_1}\ket{e_2}^{\otimes l_2}\nonumber\\
&&\qquad\qquad\qquad= a_{l_1+l_2,0}\,\mathcal{S}\(\ket{e_1}+\frac{a_{0,1}}{a_{1,0}}\ket{e_2}\)^{\otimes (l_1+l_2)}.\nonumber
\ee
An analogous argument leads to
\be
 &&\sum_{l_3,l_4}b_{l_3,l_4} \,\mathcal{S}\ket{e_3}^{\otimes l_3}\ket{e_4}^{\otimes l_4}\nonumber\\
&&\qquad\qquad\qquad= b_{l_3+l_4,0}\,\mathcal{S}\(\ket{e_3}+\frac{b_{0,1}}{b_{1,0}}\ket{e_4}\)^{\otimes (l_3+l_4)}.\nonumber
\ee
Replacing the indices and combining them, we obtain Eq.~(\ref{4dec}),
which completes the proof.
\end{proof}

Now we can provide:

\noindent
\begin{proof}[Proof of Proposition \ref{n4iid}]
Repeating a similar argument of Lemma \ref{symmetry-general-lemma-binominal-12-34} by exchanging $\ket{e_2}$ and $\ket{e_3}$,
we find that there exists $c_{l_1,0},d_{l_2,0},c_{0,1},d_{0,1}$ such that
\begin{eqnarray}
 \ket{\Psi} &=& \sum_{l_1,l_2}	c_{l_1,0}d_{l_2,0}\,\mathcal{S}\(\ket{e_1}+\frac{c_{0,1}}{c_{1,0}}\ket{e_3}\)^{\otimes l_1}\nonumber\\
&&\hspace{2cm}\otimes \(\ket{e_2}+\frac{d_{0,1}}{d_{1,0}}\ket{e_4}\)^{\otimes l_2},\nonumber
\label{cddec}
\end{eqnarray}
where $c_{0,0}=d_{0,0}=0$. 
Thus the coefficient $\Psi_{l_1,l_2,l_3,l_4}$ admits two different expressions. 
Equating these two for $\Psi_{l_1,0,l_3,0}$, we find
\begin{equation}
 a_{l_1,0}b_{l_3,0} = c_{N,0}
{N\choose k}
\(\frac{c_{0,1}}{c_{1,0}}\)^{k},
\label{abc}
\end{equation}
where we have used $l_1+l_3=N$. Plugging Eq.~(\ref{abc}) into Eq.~(\ref{4dec}),
we find
\begin{equation*}
 \ket{\Psi} = c_{N,0}\[\ket{e_1}+\frac{a_{0,1}}{a_{1,0}}\ket{e_2}+\frac{c_{0,1}}{c_{1,0}}(\ket{e_3}+\frac{b_{0,1}}{b_{1,0}}\ket{e_4})\]^{\otimes N},
\end{equation*}
which is an i.i.d.~pure state.
\end{proof}

In the third step, we generalize Proposition~\ref{n4iid} to the case $n\geq 4$ by induction using 
Proposition~\ref{n4iid} as the initial condition.
Namely,  we wish to show:
\begin{proposition}
 \label{unisep-lemma-psi-general}
 Let $q\geq 4$ be an integer.   If for $n=q$ the statement that a bosonic pure state is USEP if and only if it is an i.i.d.~pure state is true, then it is also true for 
$n=q+1$.
\end{proposition}
Before proving this,  we recall that any state $\ket{\Psi}\in{\cal H}_{\cal S}=[(\C^{n+1})^{\ot N}]_{\cal S}$ can be written as
\begin{equation}
 \ket{\Psi} = \mathcal{S}\sum_{j=0}^N \ket{\Psi(i,j)}\otimes y_{j}(i)\ket{e_i}^{\otimes j}
 \label{ijdec}
\end{equation}
by using $y_j(i)\in\C$ and an appropriate vector $\ket{\Psi(i,j)}\in \left[V(i)^{\otimes (N-j)}\right]_{\mathcal{S}}$, where $V(i)={\rm span}\{\ket{e_j}\}_{j\neq i}$ uses an orthonormal basis $\{\ket{e_i}\}_{i=1}^{n+1}$.
This expression is quite useful since the following holds.
\begin{lemma}
 \label{symmetry-general-lemma-sub-univsep}
Let $\ket{\Psi}$ be a USEP state. Then $\ket{\Psi(i,j)}$ in (\ref{ijdec}) is USEP in $\left[V(i)^{\otimes N-j}\right]_{\mathcal{S}}$.
\end{lemma}
\begin{proof}
Suppose that $\ket{\Psi(i,j)}$ is not USEP in $\left[V(i)^{\otimes N-j}\right]_{\mathcal{S}}$.
Then there exists a pair $(\Gamma',V')$ such that the observable part $\ket{\Psi(i,j)(\Gamma',V')}$ cannot be written by a symmetrized single term.
Here, $\Gamma^\prime$ is a partition of the number $N-j$ and $V^\prime$ is a decomposition of $V(i)$ into subspaces which are orthogonal to one another.
Now, a pair $(\Gamma,V)$ is given by 
\be
\Gamma = \Gamma' \cup \{\{N-j+1,N-j+2,\cdots,N\}\},\nonumber
\ee 
and 
\be
V= V'\cup \{\text{span}\{\ket{e_i}\}\}.\nonumber
\ee
Then the measurable part $\ket{\Psi(\Gamma,V)}$ is found to be $\mathcal{S}\(\ket{\Psi(i,j)(\Gamma',V')}\otimes \ket{e_i}^{\otimes N-j}\)$.
However, because of the property of $\ket{\Psi(i,j)(\Gamma',V')}$ mentioned above, $\ket{\Psi(\Gamma,V)}$ cannot be written by a symmetrized single term.
Since this contradicts with the assumption we started with, we conclude that $\ket{\Psi}$ is USEP in $\left[V(i)^{\otimes N-j}\right]_{\mathcal{S}}$.
\end{proof}

From the assumption of induction posed in Proposition \ref{unisep-lemma-psi-general} and Lemma \ref{symmetry-general-lemma-sub-univsep}, 
we have for $n = q$,
\be
\ket{\Psi(i,j)}=\ket{\psi(i,j)}^{\ot (N-j)},\nonumber
\ee
where
\begin{equation*}
 \ket{\psi(i,j)} = \sum_{k\neq i}x_{kj}(i)\ket{e_k}\in\C^{n+1}.
\end{equation*}
Comparing the coefficients of $\ket{\Psi}$ for different $i$, we reach
\begin{lemma}
 \label{unisep-lemma-psi-ij}
The coefficients $x_{kj}(i)$ can be chosen in such a way that they are independent of $j$.
\end{lemma}
\begin{proof}
For simplicity, we take an integer $\bar{i}$ and consider the case $i=\bar{i}$.
We will show that $x_{kj}(\bar{i})$ is independent of $j$. This statement
holds for any integer $\bar{i}$, and hence we prove the lemma \ref{unisep-lemma-psi-ij}.

If the coefficients $x_{kj}(\bar{i}),y_{j}(\bar{i})$ satisfy $x_{kj}(\bar{i})y_{j}(\bar{i})=0$ for all $j,k$
such that $1\leq j\leq N, 0\leq k \leq N-j$,
then we may redefine the coefficients so that $y_j(\bar{i})=0$ for $1\leq j\leq N$.
Then the state $\ket{\Psi}$ becomes $\ket{\Psi(\bar{i}, 0)}$, which is i.i.d.

Next, we consider the case that there exist $\bar{k}$ and $\bar{j}$, 
such that $x_{\bar{k}\bar{j}}(\bar{i})y_{\bar{j}}(\bar{i})\neq 0$.
First, recall that the decomposition (\ref{ijdec}) depends on the index $i$.
Thus we have
\be
\ket{\Psi}
&=&\mathcal{S}\sum_{j=0}^N \ket{\Psi(i,j)}\otimes y_{j}(i)\ket{e_i}^{\otimes j}\nonumber\\
&=&\mathcal{S}\sum_{j=0}^N \ket{\Psi(\bar{i},j)}\otimes y_{j}(\bar{i})\ket{e_{\bar{i}}}^{\otimes j},\nonumber
\ee
where $i\neq \bar{k}$ and $i\neq \bar{i}$.
Comparing the coefficients of $\mathcal{S}\ket{e_{\bar{k}}}^{\otimes (N-\bar{j})}\ket{e_{\bar{i}}}^{\otimes \bar{j}}$ for all $i$ 
such that $i\neq \bar{k}, i\neq \bar{i}$, we find
\begin{eqnarray*}
 {N \choose\bar{j}}x_{\bar{k}0}(i)^{N-\bar{j}} x_{\bar{i}0}(i)^{\bar{j}}= x_{\bar{k}\bar{j}}(\bar{i})^{N-\bar{j}}y_{\bar{j}}(\bar{i}) \neq 0.
\end{eqnarray*}
This means $x_{\bar{k}0}(i)\neq 0$ and $x_{\bar{i}0}(i)\neq 0$.
Similarly, comparing the coefficients of $\mathcal{S}\ket{e_{\bar{k}}}^{\otimes (N-j)}\ket{e_{\bar{i}}}^{\otimes j}$ for all $j$ such that $0\leq j \leq N$, we find
\begin{eqnarray}
 x_{\bar{k}j}(\bar{i})^{N-j}y_{j}(\bar{i})={N \choose j}x_{\bar{k}0}(i)^{N-j} x_{\bar{i}0}(i)^{j} \neq 0.\nonumber
\end{eqnarray}
This means $x_{\bar{k}j}(\bar{i})\neq 0, y_j(\bar{i})\neq 0$ for all $j$.

Since $x_{\bar{k}j}(\bar{i})$ and $y_j(\bar{i})$ are nonzero, we have a well-defined ratio between the coefficients of $\mathcal{S}\ket{e_{\bar{k}}}^{\otimes (N-j)}\ket{e_{\bar{i}}}^{\otimes j}$ and $\mathcal{S}\ket{e_{\bar{k}}}^{\otimes (N-j-1)}\ket{e_{k'}}\ket{e_{\bar{i}}}^{\otimes j}$
for $i=\bar{i}$ and $i=k''$, where $\bar{k},k',k'',\bar{i}$ are different from each other.
Since these coefficients are rewritten as
\begin{eqnarray}
x_{\bar{k}j}(\bar{i})^{N-j}y_j(\bar{i})={N \choose j}x_{\bar{k}0}(k'')^{N-j}x_{\bar{i}0}(k'')^j,\nonumber
\end{eqnarray}
and
\begin{eqnarray}
&& \frac{(N-j)!}{(N-j-1)!}x_{\bar{k}j}(\bar{i})^{N-j-1}x_{k'j}(\bar{i})y_j(\bar{i})\nonumber\\
 &&\qquad\quad=\frac{N!}{(N-j-1)!j!}x_{\bar{k}0}(k'')^{N-j-1}x_{k'0}(k'')x_{\bar{i}0}(k'')^j,\nonumber
\end{eqnarray}
we obtain the ratio,
\begin{equation*}
 \frac{x_{k'j}(\bar{i})}{x_{\bar{k}j}(\bar{i})}
	= \frac{x_{k'0}(k'')}{x_{\bar{k}0}(k'')}.
\end{equation*}
Since we can choose $j,k',k''$ freely, this equation means that
the ratio of coefficients, $x_{k'j}(\bar{i})/x_{\bar{k}j}(\bar{i})$, is independent of $j$ for all $k'$.
Thus, setting $z_{k'}(\bar{i})=x_{k'j}(\bar{i})/x_{\bar{k}j}(\bar{i})$ and $w_{j}(\bar{i})=x_{\bar{k}j}y_{j}(\bar{i})$, 
we can rewrite $\ket{\Psi}$ as
\begin{equation}
\ket{\Psi}= \mathcal{S}\sum_{j=0}^N \ket{\Psi'(\bar{i},j)}\otimes w_{j}(\bar{i})\ket{e_{\bar{i}}}, \nonumber
\end{equation}
where
\be
\ket{\Psi^\prime(\bar{i},j)}=\ket{\psi^\prime(\bar{i},j)}^{\ot (N-j)}\nonumber
\ee
is given through
\begin{equation*}
 \ket{\psi^\prime(\bar{i},j)} = \sum_{k\neq \bar{i}}z_{k}(\bar{i})\ket{e_k}\in\C^{n+1}.
\end{equation*}
Using the same argument, we can show that this equation holds for all $\bar{i}$.
\end{proof}

With these, we provide:

\noindent
\begin{proof}[Proof of Proposition \ref{unisep-lemma-psi-general}]
According to Lemma \ref{unisep-lemma-psi-ij}, 
the state $\ket{\Psi}$ can be written as
\be
\!\!\!\!\!\!\!\!\ket{\Psi}&=& \mathcal{S}\sum_{j=0}^N \ket{\Psi'(\bar{n},j)}\otimes w_{j}(\bar{n})\ket{e_{\bar{n}}}, \nonumber\\
&=& \mathcal{S}\sum_{j=0}^N \ket{\Psi'(\bar{n}+1,j)}\otimes w_{j}(\bar{n}+1)\ket{e_{\bar{n}+1}}.
\label{nn+1}
\ee
If the coefficients $z_k(\bar{n}+1)w_j(\bar{n}+1)=0$ for all $j,k$ such that $1\leq j \leq N$ and $1\leq k \leq \bar{n}$,
the state $\ket{\Psi}$ becomes i.i.d.
We turn to the case that there exist $\bar{k},\bar{j}$, such that $z_{\bar{k}}(\bar{n}+1)w_{\bar{j}}(\bar{n}+1)\neq 0$.
Comparing the coefficients of $\mathcal{S}\ket{e_{\bar{k}}}^{\otimes N-\bar{j}}\ket{e_{\bar{n}+1}}^{\otimes \bar{j}}$, we find
\begin{eqnarray}
 &&{N\choose\bar{j}}z_{\bar{k}}(\bar{n})^{N-\bar{j}} z_{\bar{n}+1}(\bar{n})^{\bar{j}} = z_{\bar{k}}(\bar{n}+1)^{N-\bar{j}}w_{\bar{j}}(\bar{n}+1) \neq 0.\nonumber
\end{eqnarray}
This means $z_{\bar{k}}(\bar{n})\neq 0$ and $z_{\bar{n}+1}(\bar{n})\neq 0$.
Further comparing the coefficients of $\mathcal{S}\ket{e_{\bar{k}}}^{\otimes (N-j)}\ket{e_{\bar{n}+1}}^{\otimes j}$,
we have
\begin{eqnarray}
 w_j(\bar{n}+1)= {N \choose j}\(\frac{z_{\bar{k}}(\bar{n})}{z_{\bar{k}}(\bar{n}+1)}\)^{N}\(\frac{z_{\bar{k}}(\bar{n}+1)z_{\bar{n}+1}(\bar{n})}{z_{\bar{k}}(\bar{n})}\)^j.\nonumber
\end{eqnarray}
Substituting this expression of $w_j(\bar{n}+1)$ into the Eq. (\ref{nn+1}), we obtain
\be
\ket{\Psi}=\ket{\psi}^{\ot N}\nonumber
\ee
with
\begin{eqnarray}
  \ket{\psi}
	 =\sum_{k\neq\bar{n}+1}  \frac{z_{\bar{k}}(\bar{n})z_{k}(\bar{n}+1)}{z_{\bar{k}}(\bar{n}+1)} \ket{k}
	 +z_{\bar{n}+1}(\bar{n})\ket{e_{\bar{n}+1}},\nonumber
\end{eqnarray}
as required.
\end{proof}

This allows us to complete our proof.

\noindent
\begin{proof}[Proof of Case \ref{unisep-bosonic-n4} of Theorem \ref{th2}]
Combining Proposition \ref{n4iid} and \ref{unisep-lemma-psi-general}, 
we learn that the statement of case 2 holds by induction.
\end{proof}

\section{No Universally Entangled States}
\label{univent}

So far, we have considered only USEP states, but the opposite extreme case may also be worth studying.   Namely, we are interested in the existence of states which are entangled for any choice of measurement setups.   Analogously to USEP states, we introduce:

\begin{definition}
A state $\ket{\Psi}$ is universally entangled (UENT) if $\ket{\Psi(\Gamma, V)}$ is entangled for any $(\Gamma,V)$.
\end{definition}
Unlike USEP,
however, the notion of UENT states is actually useless  because of the following no-go theorem:
\begin{theorem}
\label{th3}
There exists no UENT states for both bosonic and fermionic systems.
\end{theorem}
\begin{proof}
First, we consider the fermionic case.
Given a basis $\{\ket{e_j}\}$, any fermionic state $\ket{\Psi}\in \mathcal{H}_{\mathcal{A}}$ can be written as Eq.~(\ref{ferexp}).
Let us relabel the basis
vectors so that $\Psi_{12\cdots N}\neq0$.   If we then choose $\(\Gamma,V\)$ as
\begin{equation}
 \Gamma_k = \{k\},
 \qquad
 V_k = \text{span}\{\ket{e_k}\},\nonumber 
\end{equation}
for $k=1,2,\dots, N$, 
we find that the projected state reads $\ket{\Psi(\Gamma,V)}=\sqrt{N}\mathcal{A}\bigotimes_{i=1}^N\ket{e_i}$, which is separable.

Next, we consider the bosonic case.  Similarly to the fermionic case, any state 
$\ket{\Psi}\in \mathcal{H}_{\mathcal{S}}$ can be expanded with a basis $\{\ket{e_i}\}$ as
\be
\ket{\Psi}=\sum_{1\le i_1\le i_2\le\dots\le i_N\le n}\Psi_{i_1i_2\dots i_N}\, {\cal S}\Bot_{k=1}^N\ket{e_{i_k}}.
\nonumber
\ee
Let $\Psi_{i_1i_2\dots i_N}\neq0$ be a nonvanishing coefficient for an integer set $i_	1\le i_2\le\dots\le i_N$. 
In general, the integer set could be degenerate in the sense that
\be
i_1=&\dots&=i_{g_1},\nonumber\\ 
i_{g_1+1}=&\dots&=i_{g_1+g_2},\nonumber\\
&\dots&\nonumber\\
i_{g_1+\dots+g_{s-1}+1}=&\dots&=i_{g_1+\dots+g_{s}},\nonumber
\ee
where $s$ and $g_1,\dots, g_s$ are positive integers such that $g_1+\dots+g_s=N$.
By introducing $G_i=\sum_{k=1}^i g_k$, we may relabel the basis vectors to replace $i_{G_k}$ with $k$.
Now we choose $\(\Gamma,V\)$ as
\begin{equation}
 \Gamma_k = \{i\}_{i=G_{k-1}+1}^{G_k},
 \qquad
 V_k=\text{span}\{\ket{e_{k}}\},\nonumber
\end{equation}
where $G_0=0$ and $k=1,\dots, s$.
Then we find $\ket{\Psi(\Gamma,V)}=\sqrt{M}\mathcal{S}\bigotimes_k \ket{e_k}^{\otimes g_k}$, which is separable.
\end{proof}

This theorem shows that, whatever the state is,  we can always find a measurement setup which cannot observe quantum correlation inherent to entanglement.  In other words, with respect to that  measurement setup, the state is not entangled.

\section{Conclusion and discussions}
\label{conclusion}
In order to discuss entanglement in identical particle systems, we need to introduce a coherent scheme in which the indistinguishability of the particles is taken into account properly.   Our scheme proposed earlier \cite{Sasaki2011} and used here is one that meets this requirement.  One of the features of our scheme is that entanglement of a given state is not determined by 
the state alone, but also by the measurement setup prepared to observe the correlations furnished by the state. 
For instance, it is possible that an identical state can be regarded as separable and at the same time entangled depending on the measurement setup used.  In view of this relative nature of entanglement, we asked the question if there exist universally separable states, {\it i.e.}, those which are separable for any measurement setups, and if so, what they are.  A similar question applies to the other extreme case of universally entangled states.   

For fermionic systems, the answer to the former question is found to be quite simple: except for some lower dimensional cases, 
there exist no such universally separable states (Theorem \ref{th1}).  
For bosonic systems, the answer is intriguing:  apart from
some lower dimensional cases, there do exist such universally separable states, which are given exclusively by  i.i.d.~pure states: no other states can be universally separable (Theorem \ref{th2}).
 We also learned that 
the universally entangled state does not exist both in fermionic and bosonic systems, irrespective of the dimension
of the one-particle Hilbert space (Theorem \ref{th3}).
We note that these results were obtained upon the assumption that the entire class of measurement setups is specified by the pair $(\Gamma,V)$.
Since a possibility of more general measurement setups has been mentioned earlier \cite{Zanardi2004, Barnum2003}, our results may require some revision if our scheme is extended to accommodate such generalization.

Theorem 1 suggests that i.i.d.~pure states occupy a privileged position in the context of entanglement
in identical particle systems.   Note that these i.i.d.~pure  states belong to a special class of i.i.d.~distributions (or mixed states).  Whereas the latter are generically classical 
and easily generated, the former are hard to generate.
One way to generate the former is to realize a Bose-Einstein condensation of non-interacting particles whose ground state is unique in zero temperature regime. 
This indicates that, in generic situations, all the states of identical particle
systems are basically entangled under some appropriate measurement setup.


\end{document}